\documentclass{article}

\PassOptionsToPackage{numbers, compress}{natbib}
\usepackage[preprint]{neurips_2026}

\usepackage[utf8]{inputenc} % allow utf-8 input
\usepackage[T1]{fontenc}    % use 8-bit T1 fonts
\usepackage{hyperref}       % hyperlinks
\usepackage{url}            % simple URL typesetting
\usepackage{booktabs}       % professional-quality tables
\usepackage{amsfonts}       % blackboard math symbols
\usepackage{nicefrac}       % compact symbols for 1/2, etc.
\usepackage{microtype}      % microtypography
\usepackage{xcolor}         % colors
\usepackage[pdftex]{graphicx}
\usepackage{amsmath}
\usepackage{amsthm}
\usepackage{physics}
\newtheorem{theorem}{Theorem}[section]

\newtheorem{lemma}[theorem]{Lemma}
\usepackage{amssymb}
\usepackage{geometry}
\usepackage{thm-restate}

\newcommand{\E}{\mathbb{E}}
\newcommand{\Var}{\operatorname{Var}}

\newcommand{\vone}{\mathbf{1}}
\newcommand{\bl}{\beta_L}
\DeclareMathOperator*{\argmin}{arg\,min}

\title{Synaptic clustering emerges from learning and supports covariance discrimination}

\author{%
  Ilenna S. Jones \thanks{Corresponding author.} \\
  Kempner Institute\\
  Harvard University\\
  Allston, Boston, MA 02134 \\
  \texttt{ijones@g.harvard.edu} \\
  \And
  Maceo D. Richards \\
  Kempner Institute \\
  Harvard University\\
  Allston, Boston, MA 02134\\
  \And
  Houman Safaai \\
  Kempner Institute \\
  Harvard University\\
  Allston, Boston, MA 02134\\
  \And
  Elom Amematsro\\
  Kempner Institute \\
  Harvard University\\
  Allston, Boston, MA 02134 \\
  \And
  Bernardo Sabatini \thanks{Corresponding author.} \\
  Kempner Institute \\
  Harvard University\\
  Allston, Boston, MA 02134 \\
  \texttt{bernardo\_sabatini@hms.harvard.edu} \\
}

\begin{document}

\maketitle

\begin{abstract}
Functional synapse clusters (FSCs) are synapses with correlated presynaptic activity that are colocalized on the same neuronal dendritic branch.  FSCs have been observed after learning in cortical and hippocampal pyramidal neurons.  However, previous efforts to ablate FSCs by pharmacologically blocking dendritic nonlinearities to establish causal necessity may have confounded effects. Therefore, whether FSCs are causally necessary for computation is unknown. Here, we attempt to isolate FSCs from this potential confounder in silico. We train Dendrinet, an artificial neural network architecture with hierarchical dendritic segments and sparse conductance-based synapses, on a Permuted-Covariance Classification (PCC) task. This task cannot be solved by single-layer linear-nonlinear artificial neural networks. We find that neurons with dendrites can be trained to solve the task and develop excitatory and inhibitory FSCs if both dendritic nonlinearities and synaptic structural plasticity are active. Turning off dendritic nonlinearities reduces excitatory FSCs, which replicates experimental findings, and reduces performance while unexpectedly increasing inhibitory FSCs. Furthermore, shuffling learned synaptic connectivity while keeping the nonlinearities fixed reduces performance. This shows sensitivity to learned connectivity, but the shuffle does not change only FSCs. Shuffling inhibitory synapse properties reduces performance more than the corresponding excitatory shuffle, showing higher sensitivity to inhibitory organization. This work suggests that dendritic compartmentalization and learned synaptic organization can support computation of covariance structure.
\end{abstract}

\section{Introduction}

Learning and memory depend on how input synapses are spatially organized onto their receiving dendritic trees \citep{iacaruso_synaptic_2017, kastellakis_synaptic_2019}. This spatial organization impacts signal processing due to the nonlinear ion-channel-derived excitability properties in the neural membrane \citep{london_dendritic_2005, stuart_dendritic_2015}. Synapses carry different sensory input tuning, and how those synapses are spatially arranged promises to reveal how neural computation occurs \citep{wilson_orientation_2016, iacaruso_synaptic_2017}.

Synaptic plasticity, via strengthening, weakening, and rewiring synapses, allows animals to adapt to their changing environment \citep{citri_synaptic_2008, holtmaat_experience-dependent_2009}. In response to synaptic plasticity, similarly tuned synapses form spatially clustered groups on dendritic trees, named functional synapse clusters (FSCs) \citep{kleindienst_activity-dependent_2011, takahashi_locally_2012, wilson_orientation_2016, kastellakis_synaptic_2019, kastellakis_dendritic_2023}. This phenomenon occurs across cortex and hippocampus, reflecting changes due to learning events and memory formation \citep{frank_hotspots_2018, lai_fear_2018, bloss_single_2018, kastellakis_synaptic_2019}. A key regulator of both plasticity and FSC formation is NMDA receptor (NMDAR) activity, which acts as a coincidence detector, activating only under sufficient local depolarization while glutamate is bound \citep{mayer_voltage-dependent_1984, nowak_magnesium_1984}. Because this depolarization must be built locally from multiple small excitatory postsynaptic potentials (EPSPs), co-activation of similarly tuned, proximal FSCs is well suited to provide it \citep{schiller_nmda_2000, major_active_2013}. In this feedback loop, supralinear NMDAR activity drives plasticity, plasticity builds FSCs, and co-active FSCs produce the local depolarization that drives NMDAR activity.

\begin{figure}[t!]
    \centering
    \includegraphics[width=1.0\linewidth]{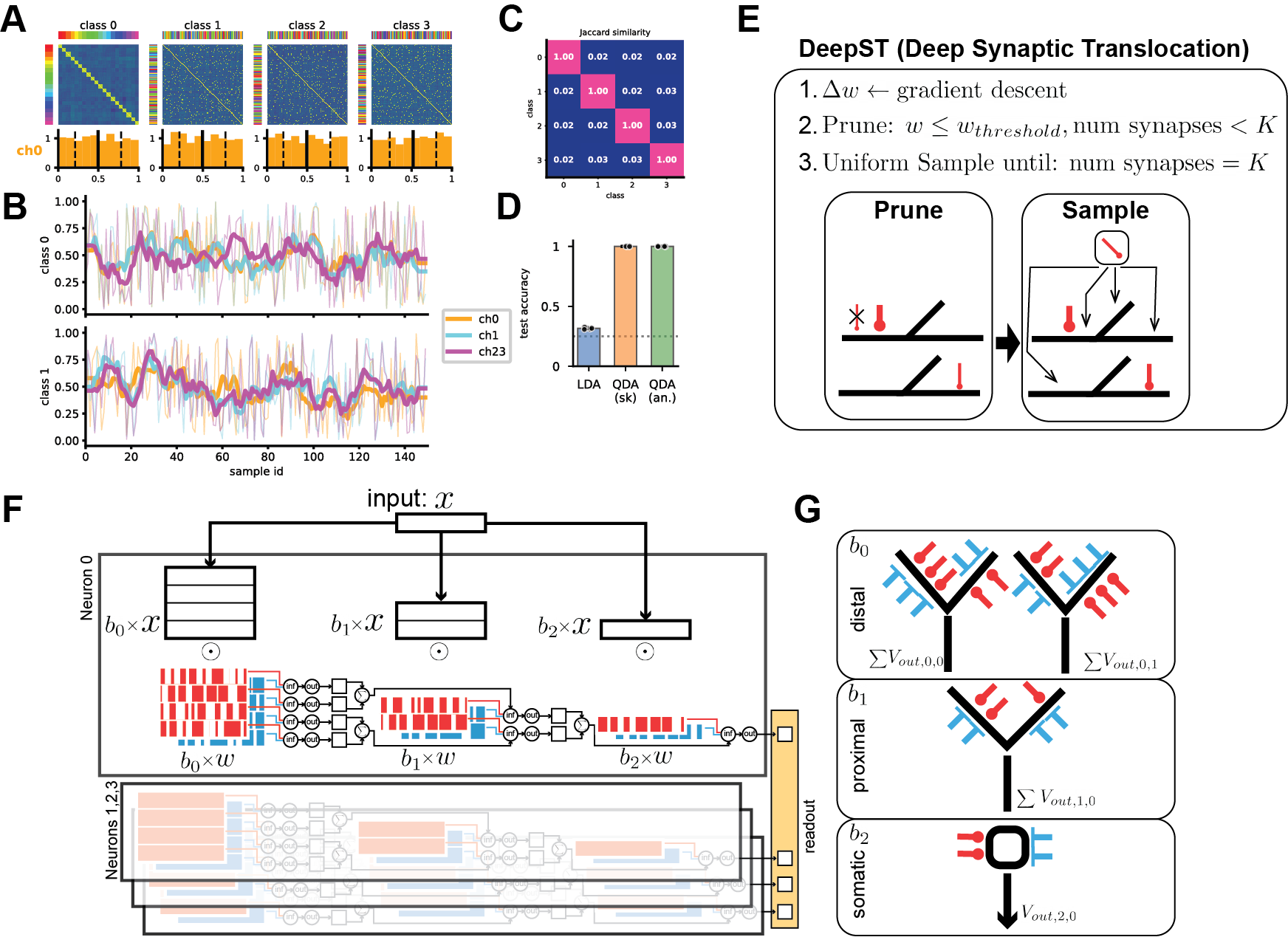}
    \caption{\textbf{Task optimization framework for dendritic network model} 
     a. Permuted Covariance Classification (PCC) Task. Each class differs by permutation of input channels, which interchanges  the rows and columns of their correlation matrices. Every channel belongs to a "block" of correlated input dimensions. No single channel has discrimination information because the first-order statistics (mean and variance) are equal across all classes. b. Samples from input channels. Channels 0 and 1 are highly correlated (r=0.9) in class 0 and have low correlation (r=0.3) in class 1. Channels 1 and 23 have low correlation (r=0.3) in class 0 and are highly correlated (r=0.9) in class 1. c. Jaccard similarity matrix shows that the classes have minimal similarity in their second-order correlation structure. d. Linear discriminant analysis (LDA) and quadratic discriminant analysis (QDA) performance shows this task is not trivially solved using a linear classifier and is solvable by quadratic discrimination. e. Deep Synaptic Translocation (DeepST) optimization rule. Weights in the network are updated with gradient descent, pruned if the weight value falls below a small threshold, and then replaced anywhere in the network until the network's fixed synapse capacity is reached. f. Dendritic neural network model "Dendrinet". Each of the four dendritic neurons in the network is a small, sparsely connected, binary tree network. g. A single neuron of the dendritic neural network model Dendrinet displayed as a cartoon neuron. Distal, proximal, and somatic branch layers of the dendritic neuron receive a sparse selection of excitatory and inhibitory inputs.}
    \label{fig:methods}
\end{figure}

Although FSCs have been found to reliably occur in response to learning and memory formation, it remains unclear whether FSCs support computations underlying these phenomena, or merely co-occur with them. Previous work attempting to test the necessity of FSCs for these events blocks formation of FSCs using Amino-phosphonovaleric acid (APV), a selective NMDAR antagonist \citep{morris_selective_1986, kleindienst_activity-dependent_2011, sehgal_compartmentalized_2025}. However, blocking NMDAR not only interrupts formation of FSCs, but also blocks supralinear NMDAR activity \citep{schiller_nmda_2000, antic_decade_2010}. Given that NMDAR-derived nonlinearities may be necessary for relevant computations \cite{poirazi_pyramidal_2003, polsky_computational_2004, beniaguev_single_2021, aizenbud_what_2026}, the observed behavioral and circuit perturbations could therefore reflect the loss of either FSCs or NMDAR activity.

To test necessity directly, we need to disrupt FSCs while leaving the nonlinearity intact. Since it is technologically infeasible to rewire synapses on dendrites \citep{holtmaat_experience-dependent_2009}, we turned to computational neuron modeling. We borrowed the task optimization framework from deep learning to build neurons that solve a task \citep{yamins_using_2016, richards_deep_2019, jones_might_2021}.  Tasks that require FSCs likely involve nonlinear interactions between input pairs and are thus sensitive to second-order structure \citep{mel_nmda-based_1992, mel_translation-invariant_1998}. Therefore, we use a task in which performance requires analysis of input covariance. For the neuron model, instead of asserting clusters a priori, we introduce biologically observed dendritic nonlinearities and a structural plasticity rewiring algorithm. This allows for the emergence of FSCs in the context of solving the task \citep{ujfalussy_impact_2020}. For the optimization algorithm, we borrow from the sparse neural network literature to introduce a simple structural plasticity rewiring rule that pairs with gradient descent optimization \citep{bellec_deep_2018}. FSCs emerge naturally during task optimization and can be ablated post-training to test their contributions to computation. 

Using this fully observable, "glass-box" neural learning apparatus \citep{hassija_interpreting_2024}, we determine how the model, when enabled by its biological properties, solves the task. The dendritic nonlinearities in the model consist of nonlinear synaptic integration with shunting inhibition and an NMDAR-derived nonlinearity. We find that these two nonlinearities and a structural plasticity rewiring algorithm are all necessary for effective performance of the task. Class-selective excitatory and inhibitory clustering appears in the distal-most branches of the dendritic tree. The trained branches show class-selective clustering and activity differences. Together these results suggest, but do not isolate, a contribution from FSCs.

\section{Methods}

\subsection{Task: Permuted Covariance Classification}

We designed the task such that class information lives in the relationships between channels rather than in the activity of any single channel, requiring the classifier to discriminate samples by their covariances as opposed to single-channel statistics (Fig. \ref{fig:methods}AB). Input samples are 200-dimensional vectors with block-diagonal covariance structure, and classes are generated by permutation of the original channel order, which permutes the rows and columns of the covariance structure for each class. To construct the classes, a base block-diagonal correlation matrix $R^{(0)}$ with N=200 input dimensions is built from 20 contiguous blocks of size 10. Input channels within the same block are highly correlated ($r = 0.9$) whereas those across different blocks have low correlation ($r = 0.3$). Per-class matrices $R^{(c)}$ for $c\in\{0, 1, 2, 3\}$ are constructed by applying class-specific permutations of input-dimension indices to $R^{(0)}$, shuffling the correlation structure across input dimensions. Each input sample has channels drawn via a Gaussian copula on $R^{(c)}$. Channels are first drawn in $z$-space, where $Z \sim \mathcal{N}(0, R^{(c)})$, and are then transformed to $x$-space via the cumulative distribution function, $X = \Phi(Z)$. As a result, each channel in $x$-space is uniformly distributed on $(0,1)$ with mean 0.5. The transform keeps the same block pattern of stronger and weaker dependence, but it does not keep the same Pearson correlation numbers. Since per-channel marginals are matched across classes, accuracy on PCC requires the model to encode pairwise feature relationships. Functional synapse clusters could plausibly support this encoding structure.

\subsection{Model: Dendrinet}
\subsubsection{Architecture}
Dendrinet is a feedforward network of dendritic neurons, with signals flowing from distal dendrite compartments toward the soma (Fig. \ref{fig:methods}FG). There are 4 excitatory output nodes acting as the classification readout, each connected to a binary branching dendritic tree. For example, a cell with a depth-2 dendritic tree has 3 branch layers consisting of 7 nodes total, with 4 distal nodes, 2 proximal nodes, and 1 soma node. This results in a Dendrinet with 4 cells that include 16 distal nodes, 8 proximal nodes, and 4 soma nodes as readouts. The excitatory readout cells receive sparse non-negative excitatory and inhibitory synaptic input from the same input sample, creating a feedforward circuit for both excitation and inhibition. Every branch node (including somatic) receives synaptic inputs and, if applicable, upstream child branch activation inputs. At initialization, the number of active synapses per branch averages 20, which is 10\% of the 200 input channels, but the exact count varies across branches under a fixed neuron-wide budget. No more than 1 synapse from a given input channel can connect onto the same branch. Synaptic weights are constrained to be positive by applying the softplus function to real-valued pre-weight parameters.

\subsubsection{Nonlinearities}
Excitatory and inhibitory synaptic integration occurs via conductance-based interactions  in the steady-state voltage equation \citep{koch_biophysics_1999}:

\begin{equation}\label{eq:vinf}
    V_\infty = f(g_{exc}, g_{inh}, g_v, V_{out}) = \frac{g_{exc} + \sum_i g_{v,i}V_{out, i}}{g_{exc} + g_{inh} + g_v + 1}
\end{equation}

where $g_E = w_E^T x$ and $g_I = w_I^T x$, such that $x$ is the input vector, and $w_E$ and $w_I$ are sparse connectivity vectors for excitation and inhibition respectively (Fig. \ref{fig:benchmark}E).  This equation arises from setting the excitatory reversal potential to 1, the inhibitory potential to 0, and the reversal potential of a fixed leak conductance (value 1) at 0.  The leak conductance can be fixed at 1 because all conductances appear as ratios.  Inhibition and leak conductances operate as "shunts" that appear only in the denominator and divisively lessen the effect of excitatory currents. Each branch also receives variable $V_{out}$ from upstream branch activations that are scaled by non-negative weight $g_v$. The output of this function produces a single branch "pre-activation" voltage between 0 and 1 that is further transformed by the Poirazi nonlinearity, which acts as an activation function (Fig. \ref{fig:benchmark}E). 

The Poirazi nonlinearity is an NMDAR-proxy function derived from \citet{poirazi_arithmetic_2003} and experimentally validated in \citet{polsky_computational_2004}. It is a piecewise linear-to-tanh function that approximates the thresholded amplifying effect of NMDAR-mediated dendritic spikes:

\begin{equation}\label{eq:vout}
V_{out} = h(V_\infty; m, b) = 
\begin{cases} 
b + (1-b) \tanh\left( \frac{ m (V_\infty - b) }{ 1 - b } \right) & \text{if } V_\infty > b \\ 
V_\infty & \text{if } 0 < V_\infty \le b \\ 
0 & \text{otherwise} 
\end{cases}
\end{equation}

The nonlinearity is rectified to zero for negative values. Hyperparameters $m$ and $b$ correspond to the slope of the tanh function and the threshold respectively. To choose the hyperparameters m and b, we used values that yielded optimal performance for a 7-node Dendrinet model with all biological components active. The output of this function produces a single branch activation between 0 and 1 that is received by downstream branches via the shunting inhibition function. These nonlinearities can be activated or inactivated for the purposes of testing their impact on model performance and the formation of functional synapse clustering.

\subsubsection{Initialization}
Each dendritic branch layer is initialized by drawing synaptic weights from a fan-in–scaled Gaussian, re-balancing the expected excitatory and inhibitory conductance of every branch, and setting each branch's resting operating point.
The pre-transform synaptic parameters for valence $v \in \{\text{exc}, \text{inh}\}$ are drawn from a Gaussian whose standard deviation follows Glorot/Xavier-normal scaling for an input dimensionality $N_v$:
\begin{equation}
\theta_v \sim \mathcal{N}\!\big(\mu_v,\ \sigma_v^2\big),
\qquad
\sigma_{\text{exc}} = \sqrt{\tfrac{2}{N_{\text{exc}}+1}},
\quad
\sigma_{\text{inh}} = \sqrt{\tfrac{2}{N_{\text{inh}}+1}}.
\end{equation}
The non-negative synaptic weight is the softplus transform of this parameter, $w = \mathrm{softplus}(\theta) = \log\!\big(1 + e^{\theta}\big)$, and only the $K$ strongest synapses per branch are retained. The excitatory mean is held at $\mu_{\text{exc}} = 0$; the inhibitory mean $\mu_{\text{inh}}$ is determined by the balancing step below. Drawing weights from this distribution alone leaves the expected excitatory and inhibitory conductance of each branch unbalanced. The expected per-branch conductance for valence $v$ is
\begin{equation}
\mathbb{E}[g_v] = \tfrac{1}{2}\, K_v\, \mathbb{E}\!\left[\,\mathrm{softplus}(X) \,\middle|\, X > \alpha_v\,\right],
\qquad
\alpha_v = \text{the } q_v\text{-quantile of } \mathcal{N}(\mu_v, \sigma_v^2),
\quad
q_v = \tfrac{K_v}{N_v},
\end{equation}
where the conditioning $X > \alpha_v$ reflects that only the strongest $K_v$ of $N_v$ synapses survive, and the factor $\tfrac{1}{2}$ is the expected fraction of active inputs. To remove the imbalance, the inhibitory mean $\mu_{\text{inh}}$ is adjusted — and the inhibitory standard deviation is scaled by the synapse-count ratio $K_{\text{inh}}/K_{\text{exc}}$ when that ratio is below 1 — until the expected excitatory and inhibitory conductance are approximately equal:
\begin{equation}
\mathbb{E}[g_{\text{exc}}] \approx \mathbb{E}[g_{\text{inh}}]
\quad\Longleftrightarrow\quad
\mathbb{E}[w_{\text{exc}}] \approx \frac{K_{\text{inh}}}{K_{\text{exc}}}\,\mathbb{E}[w_{\text{inh}}].
\end{equation}
From the balanced conductances, each branch's expected steady-state voltage and reactivation nonlinearity are initialized so the branch begins at a usable operating point. With shunting inhibition active, the expected steady-state synaptically-driven voltage (i.e., ignoring upstream branches) is
\begin{equation}
\mathbb{E}[V_\infty] = \frac{\mathbb{E}[g_{\text{exc}}]}{\mathbb{E}[g_{\text{exc}}] + \mathbb{E}[g_{\text{inh}}] + 1},
\end{equation}
which is centered near $0.5$, and the reactivation nonlinearity is initialized with a slope and a bias equal to this expected steady-state voltage.

\subsection{Optimizer: DeepST}

To allow functional synapse clusters to form, we introduced a connectivity rewiring algorithm (Fig. \ref{fig:methods}E) , Deep Synaptic Translocation (DeepST) (named after the original algorithm DEEP R for use with sparse ANNs \citep{bellec_deep_2018}). Dendrinet uses effective synaptic weights using an element-wise multiplication of $softplus(\text{pre-weight}) * \text{mask}$ to create sparse connectivity. The mask is held at fixed density for both E-to-E and I-to-E synapse weights. 

DeepST is used as follows:
\begin{enumerate}
    \item Calculate the loss
    \item Take a gradient step on the model parameters
    \item Prune: acting on the mask, drop any active synapse whose effective weight falls below a weight threshold (set to 1e-6 in for the analyses presented here) by setting the mask entry to zero.
    \item Replace: acting on the mask, replace pruned synapses by sampling new connections to restore the target active count.
\end{enumerate}

No branch is allowed to have less than one excitatory or inhibitory synapse. The rewiring is applied at each optimizer step. The number of active synapses is fixed at 140 synapses of each valence per neuron. However, the number of synapses per branch (fan-in) is not fixed and can be learned.

\subsection{Model Conditions}

\begin{table}[t!]
\centering
\renewcommand{\arraystretch}{1.5}
\begin{tabular}{|p{3.5cm}|p{3.5cm}|p{5.5cm}|}

\hline
Component & Active & Inactive \\
\hline
Shunting Inhibition & Divisive inhibition \newline (equation \ref{eq:vinf}) & Subtractive inhibition \newline $V_\infty = g_E - g_I + \sum_i g_{v,i}V_{out, i}$ \\
\hline
Poirazi Nonlinearity & Piece-wise \newline linear-to-tanh function \newline (equation \ref{eq:vout}) & Identity \newline $V_{out} = \text{Identity}(V_\infty)$ \\
\hline
Deep Synaptic \newline Translocation (DeepST) & Learned Synaptic \newline Connectivity & Frozen random synaptic connectivity \\
\hline
\end{tabular}
\caption{Model Components and Ablations}
\label{tab:model}
\end{table}

To determine whether model performance and emergence of functional synapse clusters depend on each biological component, we train and test models with each of the 3 biological components are active or inactive (Table 1). The three components are the shunting inhibition, the Poirazi nonlinearity, and the DeepST optimizer. When shunting inhibition is inactivated, these synaptic inputs are integrated linearly (i.e. $\sum  \text{excitation} - \sum \text{inhibition}$). When the Poirazi nonlinearity is inactivated, this nonlinearity is replaced with the identity function. When DeepST is ablated, the synapse locations are frozen in their initial randomly assigned location. Combinations of deactivations yields $2^3 = 8$ model conditions.

\subsection{Training}
    \textbf{Loss}: Loss was cross-entropy on the 4-class label. 
    \textbf{Optimizer}: Using Adam as the optimizer, different parameter groups had different learning rates: synapse lr = 0.1, Dendrinet branch weight lr = 0.001. Weight decay 1e-2 is applied to each group. Gradient clipping applied at 5.0. 
    \textbf{Schedule}: Up to 3,000 epochs on the train split, with early stopping at patience = 200 epochs on the validation loss. The model with the best validation checkpoint is evaluated. 
    \textbf{Batching and split}: 10,000 training samples per run, split 80/10/10 train/validation/test. The batch size was 256. 
    \textbf{Per seed data independence}: The model training seeds and dataset seeds were set equal per run, and we used 40 seeds per model condition yielding 320 trained models. \textbf{
    Evaluation}: Test accuracy on a 1,000 sample held-out set.
    \textbf{Hardware}: Nvidia A100 GPU cluster.

\subsection{Analyses}
\subsubsection{Performance Comparison}

The Dendrinet's accuracy on the PCC task is compared against three principled reference points: a linear classifier, dense multi-layer perceptrons (MLPs) at matched compute, and a "shallow" single layer vs "deep" doublet layer MLP pair to disambiguate compartmental gain from parameter gain. All baselines are evaluated on the same dataset realization used for training, with held-out test accuracy used for comparison. Linear discriminant analysis (LDA) operates as a linear lower bound, which determines whether the task is linearly separable \citep{hastie_elements_2009}. Quadratic discriminant analysis (QDA) is a second-order reference. It is Bayes-optimal for the Gaussian $z$ model, or after applying $\Phi^{-1}$ to $x$.

Two densely connected multi-layer perceptrons (dMLPs) act as parameter-size-matched and depth-matched controls. 1-layer dMLP has one hidden layer and depth of 2, and 2-layer dMLP has 2 hidden layers and depth of 3. For depth L, the trainable weight count is $P = d_{\text{in}}H + (L-1)H^2 + Hd_{\text{out}}$, where $d_{\text{in}}$ and $d_{\text{out}}$ are the input and output dimensions, and the hidden width H is set by a search to the smallest integer that meets the matching target. Biases are excluded from the parameter count. One sparsely connected MLP (sMLP) acts as a parameter-size-matched and sparsity-matched control. It has one sparsely connected hidden layer with connection density of 0.1 of the input dimension (20 weights per node), which is connected to another sparsely connected readout layer. The sparsity is produced using a dense-MLP scaffold with a fixed random binary mask applied multiplicatively in the forward pass. This mask is not learned. Hidden layer bias is kept dense, and readout has no bias. MLP Parameter count match is applied such that each MLP variant is matched 1-to-1 to a Dendrinet of the same compartment count (1, 3, 7, 15 compartments). For each pair of models, H is found by search to the smallest hidden width that satisfies $P_{\text{active}}(H) \geq P^{\text{Dendrinet}}_{\text{eff}}$ . This means MLPs receive slightly more parameters than their Dendrinet match. Results are summarized as both accuracy-vs-condition bar plots and as parameter-count-vs-accuracy scatter plots.

\subsubsection{Connectivity Structure Histograms}
The Dendrinet mask depicts the learned connectivity of the network, so to visualize connectivity we created histograms of both the input distribution (synapses per input channel) and the "fan-in" distribution (synapses per branch). This was determined for each of the 8 model conditions and each histogram was pooled over 40 seeds. The input distribution reveals whether the trained mask covers the input uniformly or concentrates on specific dimensions. The fan-in distribution reveals whether the synapses are distributed evenly at approximately 20 synapses per branch type (distal, proximal, somatic).

\subsubsection{Correlation of Excitation and Inhibition}
We determined the functional correlation of excitatory and inhibitory activity at each branch for each class context. For each branch in each trained model, we ran a test batch consisting of 100 samples $\times$ 4 classes through the network, then extracted the per-branch excitation $E_b$ and inhibition $I_b$ and computed the Pearson correlation $r(E_b, I_b)$ within each class. We then averaged across branches for each neuron and each branch layer, and aggregated the mean across the 40 seeds.

\subsubsection{Normalized Entropy Clustering Metric}
The functional synapse clustering metric is a per-branch measure of how much a branch's input synapses concentrate onto the task's covariance blocks. It is corrected for fan-in by a fan-in–matched random null. The metric reduces to one scalar per (branch, class, valence) and is directly comparable across branches, layers, and conditions.
\textbf{Setup}: The PCC task partitions the 200-dimensional input into 20 blocks of 10 dimensions under a class-specific permutation. For each class $c$, every input dimension $i$ carries a block label $g_c(i) \in \{1, \dots, 20\}$. The trained branch is characterized by its synaptic weights, so the metric summarizes where the branch's input \emph{weight mass} lands relative to class $c$'s blocks.
\textbf{Observed entropy}: For branch $b$ under class $c$, let $\mu_{b,c,g}$ be the total weight of the branch's synapses falling in block $g$ (and $0$ if none do):
\[
\mu_{b,c,g} = \sum_{i\in b:\,g_c(i)=g} w_i.
\]
These block-total weights are normalized into a distribution,
\begin{equation}
p_{b,c,g} = \frac{\mu_{b,c,g}}{\sum_{g'} \mu_{b,c,g'}},
\end{equation}
from which the observed Shannon entropy per branch and class is
\begin{equation}
H_{\text{obs},\,b,c} = -\sum_{g} p_{b,c,g}\,\log\!\big(p_{b,c,g} + \varepsilon\big),
\end{equation}
with $\varepsilon$ a small constant for numerical stability. Low $H_{\text{obs}}$ means the branch concentrates its weight on a few class-$c$ blocks; high $H_{\text{obs}}$ means it spreads weight mass across many class-$c$ blocks.
\textbf{Null entropy}: The null entropy $H_{\text{eff}}(\kappa_b, c)$, where $\kappa_b$ is the branch fan-in, is computed by Monte Carlo sampling. For each of $M = 2{,}000$ draws, $\kappa_b$ input indices are sampled uniformly without replacement from $\{0, \dots, 199\}$; the branch's observed weight values are reassigned to these locations; the block-total weights $\mu_g^{(s)}$ are recomputed under class $c$'s partition; and the same Shannon entropy formula is evaluated:
\begin{equation}
H_{\text{eff},\,b,c} = \frac{1}{M}\sum_{s=1}^{M}\left[-\sum_{g} p_{g}^{(s)}\,\log\!\big(p_{g}^{(s)} + \varepsilon\big)\right],
\qquad
p_{g}^{(s)} = \frac{\mu_g^{(s)}}{\sum_{g'}\mu_{g'}^{(s)}}.
\end{equation}
Because the draw keeps the branch's actual weight values but randomizes their placement, $H_{\text{eff}}$ is computed per branch. This null represents the case where no learning has occurred and only the fan-in is matched. The Monte Carlo null is required to prevent a fan-in confound: a branch with few synapses can only occupy a few blocks and therefore appears artificially clustered (low $H_{\text{obs}}$) regardless of learning.
\textbf{Normalized entropy}: Since observed entropy is confounded by branch fan-in, the final metric normalizes against the expected entropy under fan-in–matched random synapse placement:
\begin{equation}
H_{\text{norm},\,b,c} = \frac{H_{\text{obs},\,b,c}}{H_{\text{eff},\,b,c}}.
\end{equation}
$H_{\text{norm}} < 1$ indicates the branch is more clustered than the fan-in–matched random null; $H_{\text{norm}} = 1$ matches the null; and $H_{\text{norm}} > 1$ indicates the branch is more dispersed than the null. This is reported per condition, seed, valence, branch layer, branch, and class, with means taken across seeds.

\subsubsection{Synapse Shuffle Ablation}
The ablation test examines whether the trained branch organization is important for accuracy by applying three post-hoc shuffles to fully trained models and measuring the accuracy drop. All shuffles preserve the network's parameter count. Each synapse carries two properties: weight and connectivity mask. Each shuffle is defined by which subset it preserves vs shuffles. Each trained model is deep-copied, the shuffle is applied five times to excitatory, inhibitory, or both valences, and test accuracy is measured.

\begin{figure}[t!]
    \centering
    \includegraphics[width=1.0\linewidth]{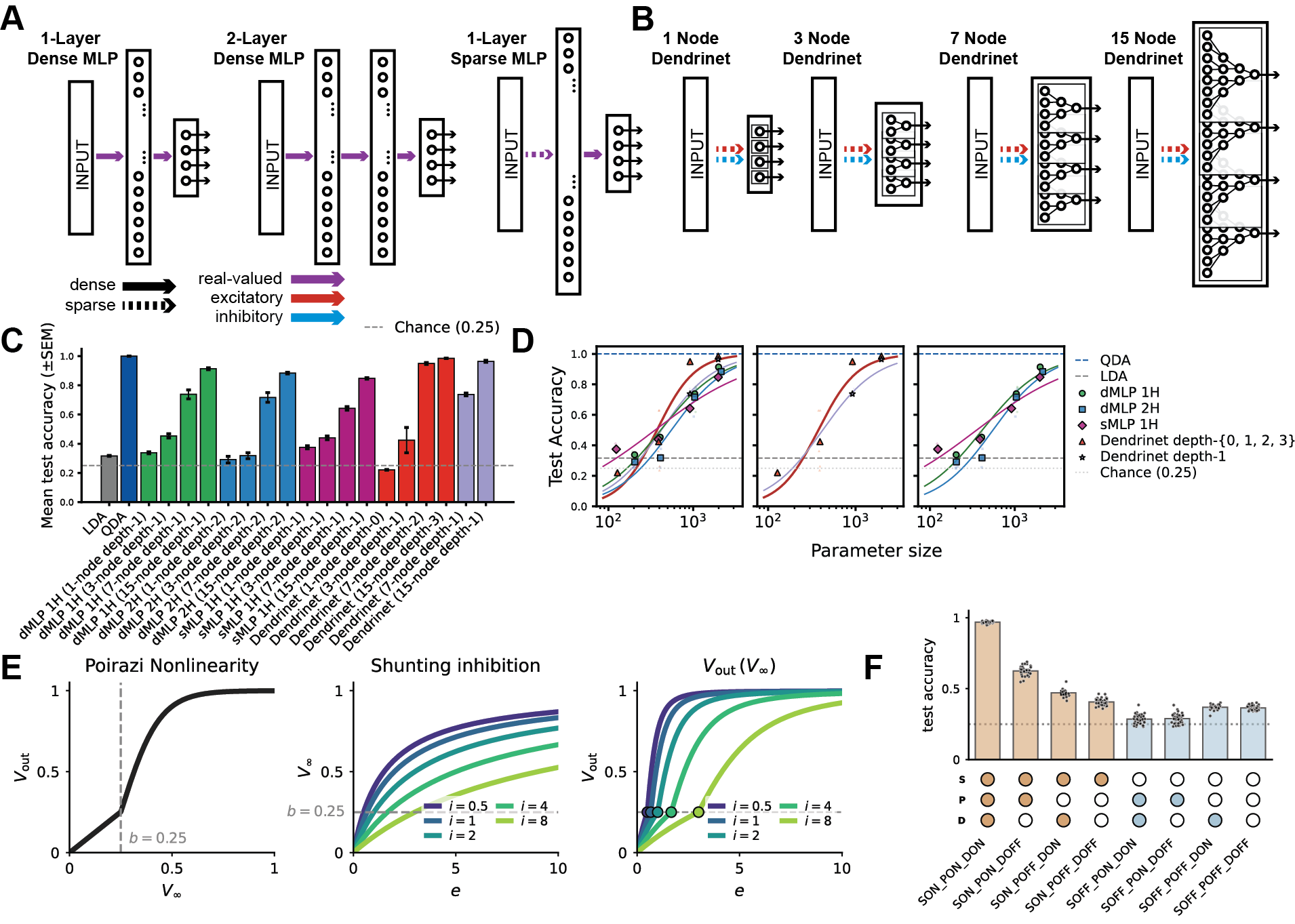}
    \caption{\textbf{Performance and parameter efficiency in Dendrinet compared to simple MLP}. 
     a. Multilayer perceptrons (MLPs) tested on PCC, including a 1-hidden layer and 2-hidden layer dense MLP and 1-hidden layer sparse MLP. All MLPs are parameter-size matched to the Dendrinet models. b. Dendrinet models with binary tree structure with increasing depth. 1 node Dendrinet is similar to a zero-hidden layer sparse MLP. Each Dendrinet has separate non-negative excitatory and inhibitory parameters, the total of which is used for parameter-size matching with MLPs. c. Model performance on PCC task for LDA, QDA, dense MLPs, sparse MLPs, and Dendrinet models. Flat Dendrinet models conserve the number of nodes while reducing depth to that of a 3 node Dendrinet model. Dendrinet (7 node) has higher performance than all other 7 node conditions. d. Model performance on PCC task, visualized for parameter scaling. The flat Dendrinet model performs similarly to the dense MLPs, whereas the tree Dendrinet with full depth outperforms all other models before performance saturates at high parameter counts. e. Dendritic nonlinearities are composed to yield an activation function conditioned on excitatory and inhibitory activity. The dendritic Poirazi nonlinearity, a biologically inspired proxy for NMDA-receptor dependent  transformation of synaptic drive, is a piece-wise linear-to-tanh amplifying function. Shunting inhibition is modeled after steady-state voltage due to open conductances with different reversal potentials. With non-negative excitation and inhibition, excitatory drive is attenuated by inhibitory divisive normalization. Each dendritic branch composes the shunting inhibition within the Poirazi nonlinearity, creating an activation function sensitive to inhibition. f. Performance of Dendrinet on the PCC task with shunting inhibition, Poirazi nonlinearity, and DeepST rule turned on or off. All three properties are necessary for performance of the PCC task above chance and above linear classifier performance.}
    \label{fig:benchmark}
\end{figure}

\section{Results}

\subsection{Dendrinet models with biological components outperform simple Multi-Layer Perceptrons in covariance task}

We benchmarked the performance of Dendrinet on the PCC task against other models. We use dense and sparse multi-layer perceptrons (MLPs) that are parameter-size matched to different sizes of Dendrinets for comparison (Fig. \ref{fig:benchmark}AB), as well as linear discriminant analysis (LDA) and quadratic discriminant analysis (QDA). As expected for this inherently nonlinear task, LDA performs very poorly ($\sim 35\%$), likely due to the lack of mean separation in the data (Fig. \ref{fig:benchmark}C). Because the task's discriminatory information lies in its across-channel second-order statistics, QDA solves it with 100\% accuracy (Fig. \ref{fig:benchmark}C). The PCC task is therefore nontrivial but solvable by second-order computation.

Dendrinet can solve this nontrivial task. We compared Dendrinet against depth- and parameter-matched MLPs, bounded below by LDA and above by QDA (Figure \ref{fig:benchmark}C). Eliminating the dendritic morphology (Dendrinet point readout neuron, the "1 Node Depth-0 Dendrinet") reduces performance relative to the full model but outperforms linear models with both more parameters and depth (Fig. \ref{fig:benchmark}CD). Dendrinet performance approaches Bayes-optimal performance at high depth and parameter counts. With reduced depth (Dendrinet depth-1), Dendrinet performance drops to that of the parameter-matched dense MLPs. Dendrinet's internal depth improves parameter efficiency before performance saturates, allowing it to outperform all other simple models at matched parameter counts.

It is unclear which of the differentiating properties of Dendrinet explain its parameter-efficient performance. Therefore, we ablated the differentiating properties of Dendrinet to determine their necessity for task performance: we replaced the Poirazi nonlinearity with an identity function, the divisive operation of the shunting inhibition with a subtractive operation, and the structural plasticity rule DeepST with fixed sparse random connectivity (Fig. \ref{fig:benchmark}F). Using each ablation individually or in combination, turning off any one of the properties drops performance closer to chance. Therefore, all three biological properties are necessary for the Dendrinet model to perform the task.

\begin{figure}[t!]
    \centering
    \includegraphics[width=1.0\linewidth]{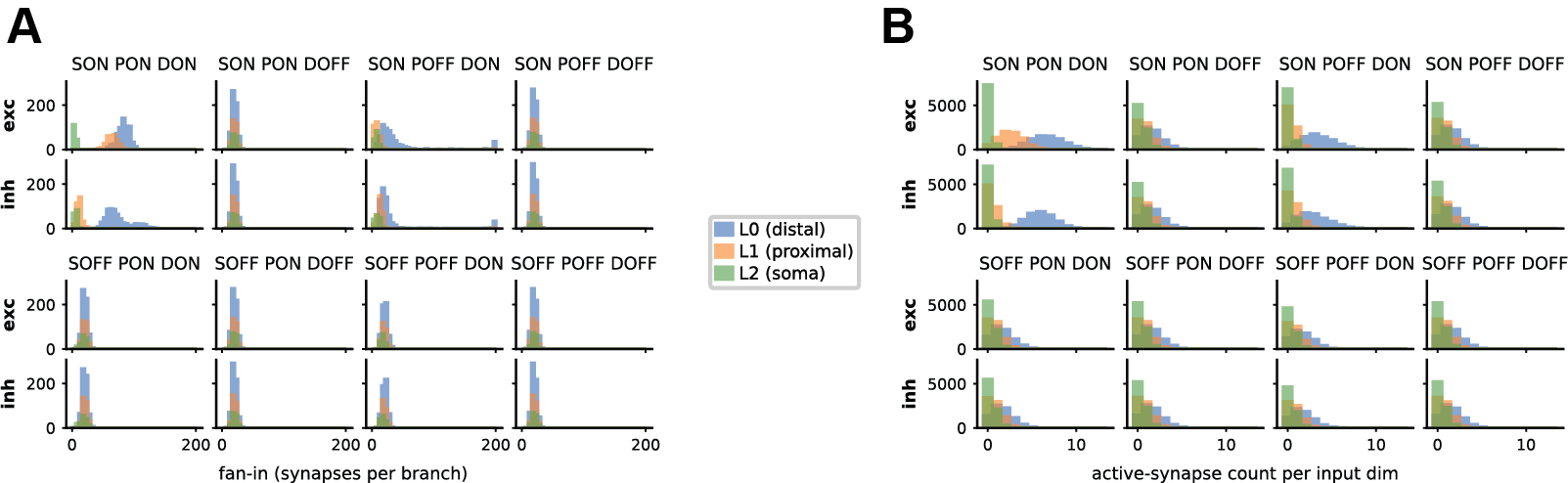}
    \caption{\textbf{Dendrinet with all biological properties active exhibits learned connectivity after task optimization} 
    a. Distribution of synapses per branch, grouped by morphological branch type. Excitatory and inhibitory organization of synaptic branch targeting sees the most change when all three properties (shunting inhibition, Poirazi nonlinearity, and DeepST) are active. Distal branches have the most excitatory and inhibitory synapses. Somatic branches have reduced synapse number. Proximal branches have more excitatory synapses than inhibitory synapses. b. Visualizing synaptic repetition through distribution of synapses per input channel, grouped by morphological branch type. When shunting inhibition and DeepST rewiring are active, specific input channels routed through both excitatory and inhibitory synapses are repeated many more times among distal branches. When the Poirazi nonlinearity is active, excitatory synapses that target proximal branches also see an increase in repeated inputs. Blue is distal, orange is proximal, and green is somatic. N=40.
    }
    \label{fig:connectivity}
\end{figure}

\subsection{Dendrinet learns connectivity structure depending on biological components necessary for task performance}

We examined if the DeepST algorithm leads to functional clustering of its synapses on the dendritic trees via structural plasticity. We can visualize connectivity by grouping neuron branch types by position in the neuron's morphology (distal, proximal, and somatic) and plotting a histogram of how many synapses each branch type has (Fig. \ref{fig:connectivity}A). In cases where DeepST was inactive, the means for each branch type are the same, reflecting the fixed random initial connectivity of the model. With DeepST active but shunting inactive the means changed slightly compared to when shunting was active. In cases where DeepST, shunting, and the Poirazi nonlinearities are all active, the distal branches gained more synapses and the somatic branches lost synapses. This occurred for both excitatory and inhibitory synapses. When the Poirazi nonlinearity is inactive, we observe the same pattern but to a much lesser extent. Similar patterns are seen when visualizing learned synaptic repetition in Figure \ref{fig:connectivity}B. This learned organization arises due to structural plasticity mechanisms in conjunction with shunting inhibition, and to a particular extent with the Poirazi nonlinearity. This demonstrates that connectivity is learned in a way that changes the number of synapses per branch. 

\subsection{Functional Synapse Clusters only form with biological components necessary for task performance}

As seen in Figure \ref{fig:benchmark}C, QDA achieves 100\% performance on the PCC task. This suggests that QDA is a useful reference for Dendrinet computation. Below is a theorem for the linearized expected SNR model based inspired by QDA.

\begin{restatable}
[Synaptic Clustering]{theorem}{synapticclustering}\label{thm:clustering_normative}

Let $0\leq\rho_g<\rho_l<1$. Let $W$ have exactly $s$ active binary synapses, with $2\leq s\leq n$, and let competing classes be independent uniform block permutations. Define the expected-competitor SNR by
\[
\bar\sigma(W,k)=\frac{W^\top\Sigma_kW}{\mathbb E_{k'}[W^\top\Sigma_{k'}W]}.
\]
Putting all $s$ synapses in one class-$k$ block gives a larger $\bar\sigma$ than putting $s-1$ there and one in another block. The difference is
\[
\Delta\bar\sigma
=\frac{2(s-1)(\rho_l-\rho_g)}{D_s}>0,
\qquad
D_s=(1-\rho_l)s+\rho_gs^2
+(\rho_l-\rho_g)\left[s+s(s-1)\frac{n-1}{N-1}\right].
\]
For fixed $s\leq n$, this expected SNR is maximal when all $s$ synapses target one block. The normalized gap is positive, but it does not always grow with $s$.
\end{restatable}

This means FSC-like clustering is theoretically beneficial in this simplified model and performance is sensitive to learned synaptic organization. The proof can be found in the Appendix.

Therefore, we hypothesized that the biology-inspired features of Dendrinet, along with the non-linear task demands, lead to FSCs. We defined FSCs as synapses with correlated activity that are located on the same dendritic branch. Exploiting the correlation structure of the task, correlated synapses are those that belong to a "block" or "assembly" of correlated input channels \citep{yuste_neuronal_2024}. To quantify the randomness in the distribution of synapses, entropy is calculated from the block distribution within each branch for each class context and normalized by an estimated entropy using Monte Carlo sampling given the number of synapses that exist on the branch. 

Based on this normalized entropy metric, we observed functional synapse clustering for both distal excitatory and inhibitory synapses in the ALL-ON condition (Fig. \ref{fig:clusters}A). Excitatory synapses were clustered in proximal branches as well. In contrast, excitatory synapses from non-preferred classes in distal dendritic compartments appeared to be "anti-clustered" or dispersed.

Connectivity learning may also correspond to FSC formation. As we've seen in Figure \ref{fig:connectivity}A for SON PON DON (all active), synapses are rewired away from the soma to the proximal and distal dendrites, depending on synapse valence. This would have an impact on clustering given the number of inputs in a block is smaller than what is connected onto a branch. However, what we see is excitatory and inhibitory clustering in the distal branches, and excitatory clustering in the proximal branches (Fig. \ref{fig:clusters}A. Comparing the connectivity result and the clustering result, we can see that the branches with more synapses also exhibit FSCs.

Ablation of any of the biologically-inspired components affects synapse clustering. In the DeepST-OFF condition (Fig. \ref{fig:clusters}B), distal excitatory clusters disappear, but proximal class-selective clusters remain. Somatic nodes show a reduction in normalized entropy (Fig. \ref{fig:clusters}B); however, it is not selective for class. In the Poirazi-OFF condition (Fig. \ref{fig:clusters}C), distal excitatory synapses show non-preferred clustering and preferred dispersal.  For both DeepST OFF and Poirazi-OFF conditions (which is Shunting ON for both), inhibitory clustering continues to be class selective, but now also shows dispersal for non-preferred classes. With Poirazi-OFF, inhibitory synapses show preferred clustering and non-preferred dispersal for every branch type (distal, proximal, somatic). For the Shunting OFF condition (Fig. \ref{fig:clusters}D), non-class-selective distal and proximal dispersal can be seen. Excitatory non-preferred class clustering and inhibitory non-selective class clustering may also be observed. Given that these modes see reduced performance, this may be due to an inability to learn the task or have the computational capacity to generalize to the test set. Since excitatory clustering was dispersed in each ablation condition, this suggests that excitatory clustering supports high performance.

\begin{figure}[t!]
\includegraphics[width=1.0\linewidth]{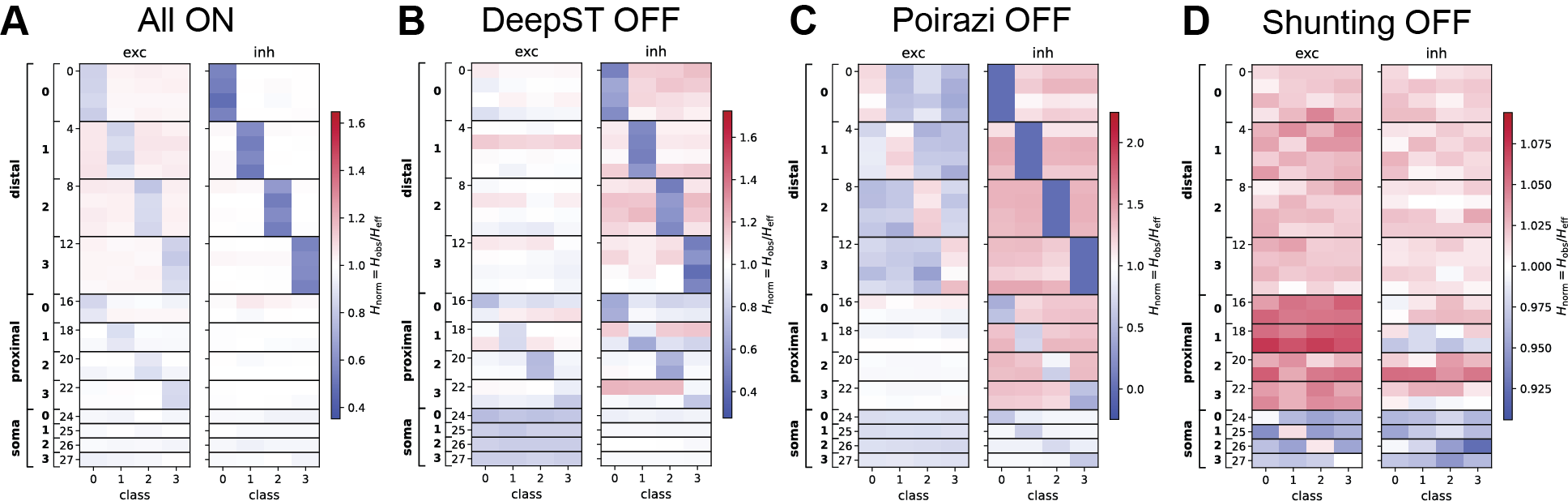}
\caption{\textbf{Normalized entropy as a metric of functional synapse clustering reveals different excitatory clustering signatures. } 
Normalized entropy yields a task-relevant organization of synaptic weights into functional synapse clusters for both distal and proximal excitatory synapses and distal inhibitory synapses. Normalized entropy is derived from block correlation connectivity distributions and is normalized by expected entropy calculated from a distribution generated via Monte Carlo sampling. Entropy less than one indicates clustering more than chance. a. With all toggled properties active, clear task-relevant excitatory and inhibitory clustering appear for the preferred class contexts. b. With DeepST deactivated, random connectivity shows distal excitatory clusters disappear, but distal inhibitory synapses remain. c. With the Poirazi Nonlinearity deactivated, excitatory clusters occur for non-preferred class contexts and "anti-cluster" or disperse for preferred class contexts. Inhibitory clustering is maximized as well. d. With shunting inhibition deactivated, only somatic, non-preferred clustering occurs among excitatory synapses and indiscriminate inhibitory clustering occurs as well. N=40. }\label{fig:clusters}
\end{figure}

%%% New section?

\subsection{Dendrinet exhibits branch activations advantageous for performance}

To interpret what the model learned through training, we visualized the branch activations, similar to examining the voltages in all sections of the dendritic trees, for each of the four readout neurons. In the ALL-ON case (Fig. \ref{fig:activity}A), distal branches receive most synaptic activation for both excitatory and inhibitory valences, with no class-specific differences in activity. We calculated the sample-by-sample correlation of excitation and inhibition data points for each branch and each class, which revealed class-specific correlation signals in the distal branches of each readout neuron. This correlation is not uniform across classes. For non-preferred classes, inhibition tracks excitation closely, so shunting inhibition cancels the excitatory drive and the branch output stays low - this is superficially similar to what is expected from concepts of excitation-inhibition balance in experimental neuroscience \citep{froemke_plasticity_2015}. For the preferred class, the correlation is weaker, leaving excitation partly un-cancelled. The residual preferred-class drive is what survives and amplified downstream (Fig. \ref{fig:activity}A). Selectivity therefore arises not from inhibition matching excitation everywhere, but from it matching most tightly for the classes that should be suppressed. Shunting inhibition and the Poirazi nonlinearity function outputs for each branch therefore cleanly receive excitation for preferred classes, and together amplify class-selective excitation.

\begin{figure}[t!]
    \centering
    \includegraphics[width=1.0\linewidth]{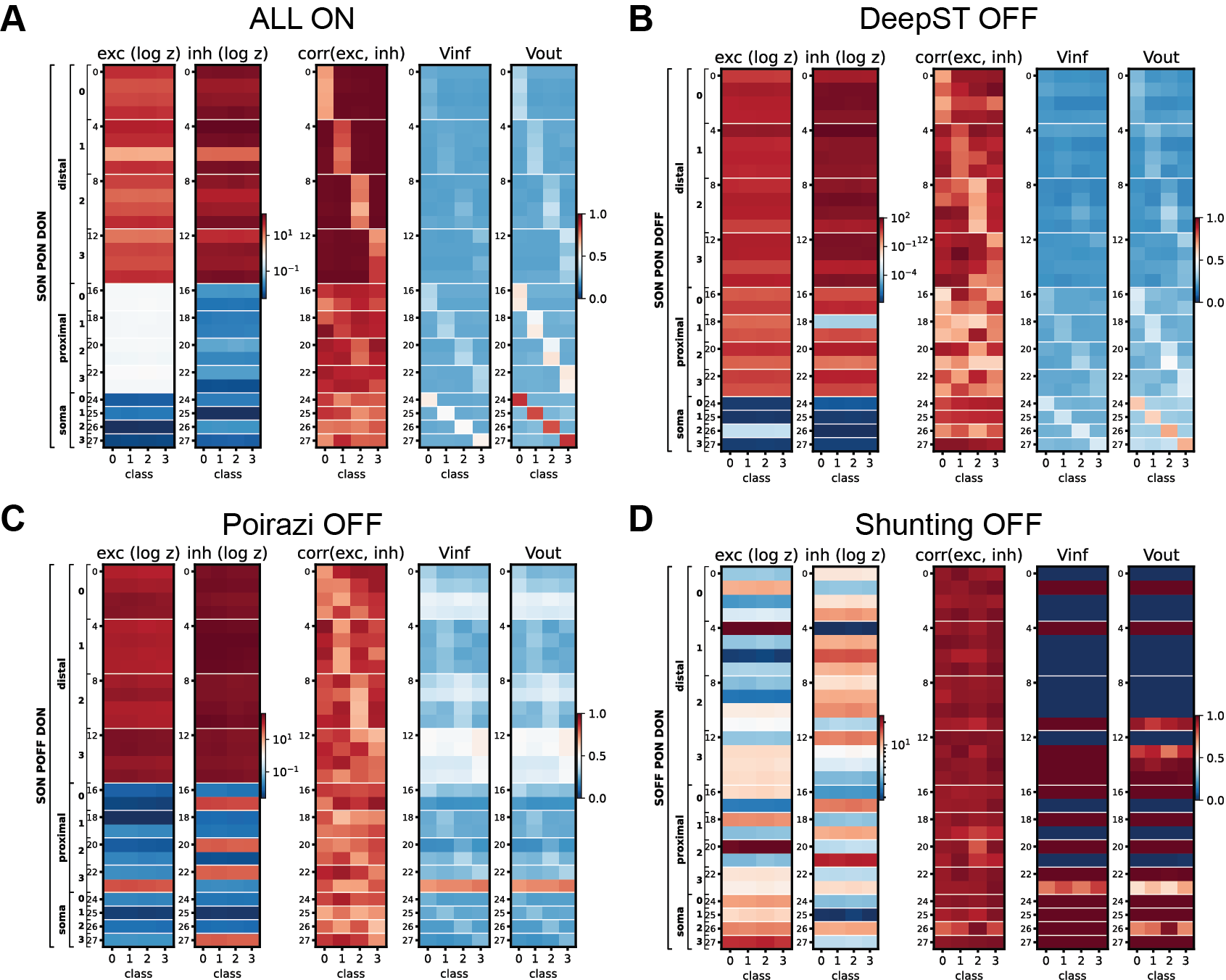}
    \caption{\textbf{Dendrinet branch-targeted synapse correlation, shunting inhibition activation, and branch output.} 
     a. Dendrinet branch activity with all three selectable biological properties active. Encoding of class specific information is mediated by excitatory and inhibitory correlation by class. Excitatory and inhibitory synaptic activation encode class-discriminatory information through high correlation of non-preferred classes and decorrelation of preferred classes in distal branches. This class-selective signal appears in distal $V_\infty$ activation, which is then amplified by the Poirazi nonlinearity that then results in accurate class-specific readout. b-d. Branch activations with one of each of the biological properties deactivated. b. DeepST OFF (random connectivity) has higher correlation in non-preferred classes, resulting in a decodable output signal at soma. c. Poirazi OFF (identity activation function) also shows slightly higher correlation in non-preferred classes, but to a much lesser degree than when the Poirazi nonlinearity is active. d. When shunting nonlinearity is OFF, inhibitory and excitatory synapses are fully correlated and no clear class-dependent signals are present. N=40.
    }
    \label{fig:activity}
\end{figure}

% Following proofs in the appendix, we also find it is optimal that the inhibitory synapses and excitatory synapses target mutually exclusive input blocks. 

% \begin{theorem}[Excitatory-Inhibitory Diversity]\label{thm:shunting_ei}
% The signal to noise ratio of a shunting neuron is maximized when the excitatory and inhibitory weights target diverse blocks of the class specific covariance matrix.
% \end{theorem}

% This means that this disjoint targeting allows for 

We hypothesized that the Poirazi nonlinearity is necessary because it converts the learned matching between excitation and inhibition into class-selective amplifications. This predicts that shunting inhibition should produce a class-dependent variance separation at the distal branches with smaller output variance for non-preferred classes than for preferred ones. We find this separation in the All-ON, Poirazi-ON case (Fig. \ref{fig:poirazi}A). This also predicts that this variance separation should align with the amplifying domains of the nonlinearity. Overlaying the shunting inhibition output onto the Poirazi nonlinearity confirms that non-preferred-class output falls into the linear-subthreshold domain, while part of the preferred-class output reaches the amplifying suprathreshold domain. These shunting inhibition outputs are then amplified by the Poirazi nonlinearity, introducing a mean shift that propagates through the network, resulting in clear class-selectivity. In the Poirazi-OFF case, the learned perfect correlation between distal excitation and inhibition disappears, and the separation between the subthreshold non-preferred classes and suprathreshold preferred class disappears, resulting in low class-selectivity. This demonstrates how the Poirazi nonlinearity acts as a clear mechanism for learned performance of the task.

\begin{figure}[t!]
    \centering
    \includegraphics[width=1.0\linewidth]{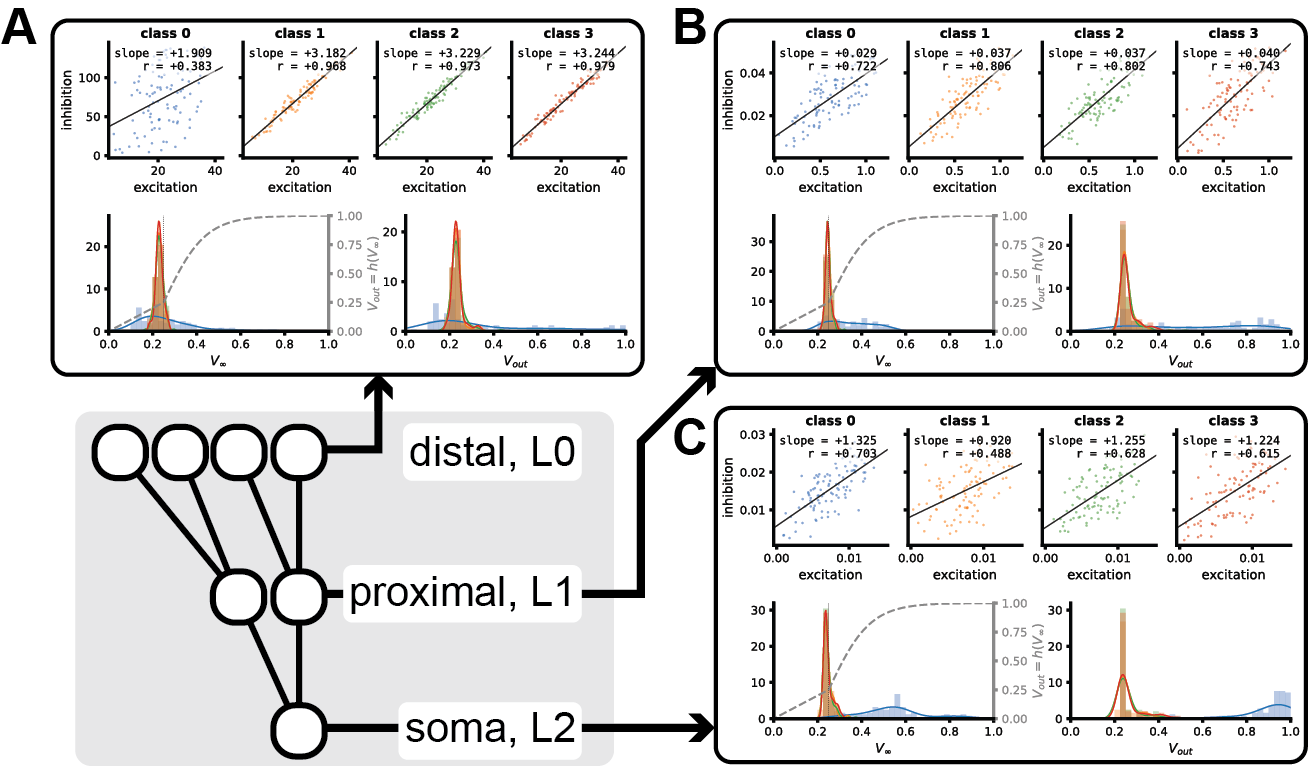}
    \caption{\textbf{Correlation between excitation and inhibition precisely allows Poirazi nonlinearity to amplify class-selective signals. } 
     a. Variance of the distribution of shunting inhibition output in the Poirazi-ON case in the distal branches are small for non-preferred classes and larger for preferred classes. When these distributions are overlaid with the Poirazi nonlinearity, the shunting inhibition output for the non-preferred classes remains in the linear-subthreshold domain. Part of the output for the preferred class lands within the amplifying supra-threshold domain. These shunting inhibition outputs are then amplified by the Poirazi nonlinearity, introducing a mean shift. Blue is distributions of activity from class 0 samples. Orange, green and red are distributions from classes 1, 2, and 3. They overlap completely to form brown. b-c. This mean shift propagates through the network, resulting in class-selectivity. The mean shift is greater in magnitude than inhibitory synaptic inputs in both proximal and somatic branch nodes, allowing excitatory signal from distal branches to be amplified at both layers.
    }
    \label{fig:poirazi}
\end{figure}

\subsection{Ablation of Functional Synapse Clusters via Shuffling Reduces task performance}

Given that FSCs emerge from optimization of Dendrinet to the PCC task, they may be essential for good performance of the trained networks. If FSCs support task performance, then shuffling learned connectivity would significantly reduce model performance. Each synapse of an FSC consists of a synaptic weight, $\hat{w}$, and a connectivity mask, $m$. An active synaptic weight is $w = m \times \hat{w}$.

We found that shuffle ablation significantly reduces performance (Fig. \ref{fig:shuffle}). Shuffling both excitation and inhibition drops performance significantly. Shuffling excitatory weights on a model with learned connectivity leads to a smaller performance drop than shuffling the excitatory connectivity mask. Shuffling the excitatory connectivity mask leads to performance similar to that of LDA. This shows that performance is sensitive to learned excitatory connectivity, but it does not separate the FSC part from the rest of connectivity. In addition, shuffling the inhibitory weights on the model drops the performance further than that of the excitatory weights. Shuffling inhibitory connectivity drops the performance even further to chance level. These support that the learned connectivity across the dendritic branches within each neuron is critical for performance of the PCC task, though performance is more sensitive to shuffling inhibitory synapse properties than excitatory ones.

\begin{figure}[t!]
    \centering
    \includegraphics[width=1.0\linewidth]{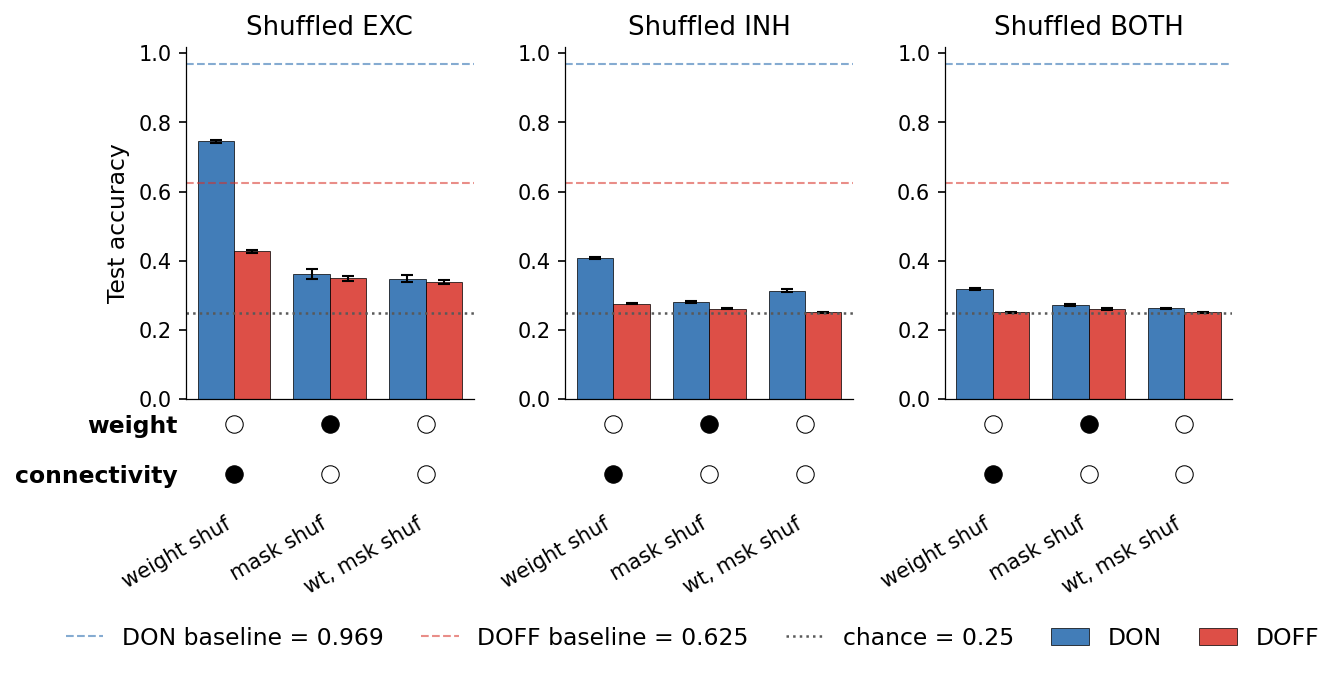}
    \caption{\textbf{Shuffle ablation of excitatory and inhibitory synapse clusters significantly reduces performance.}
    Shuffling synaptic components (weight and connectivity) disrupts task performance across the dendritic network. Excitatory and inhibitory synapses are shuffled alone, then both together. Shuffling excitatory or inhibitory connectivity distributions maximally reduces performance. Learned weight distributions also significantly reduce performance when shuffled, but to a lesser extent. Shuffling inhibitory synapses reduces performance to a greater extent than excitatory synapse shuffling. "DON" is DeepST active. "DOFF" is DeepST inactive. N = 40.}
    \label{fig:shuffle}
\end{figure}

\section{Discussion}

We examined if and how a dendrite-based model with conductance-based synapses can solve a covariance classification task and if task optimization induces FSCs. We find that task optimization produces excitatory and inhibitory FSCs on the distal branches of each readout neuron, and that fromation of these clusters depended jointly on the Poirazi nonlinearity, shunting inhibition, and structural plasticity. Furthermore, we observed correlated excitatory and inhibitory synaptic activity in distal branches in non-preferred classes, constraining the variance of the shunting inhibition function activity to remain in the linear-subthreshold domain of the Poirazi nonlinearity. In contrast, in the preferred-class context, reduced correlation of excitatory and inhibitory activity introduced variance in the shunting nonlinearity function, allowing for class-selective signals to cross a threshold and be amplified by the Poirazi nonlinearity. Directly ablate FSC by shuffling connectivity decreased performance, with performance returning to that of the single compartment baseline in the case of excitatory shuffling or chance in the case of inhibitory shuffling.

\subsection*{Dendritic computation and the PCC task}
A recurring goal in dendritic modeling is to determine if a neuron with active dendrites is more powerful computationally than a point neuron. Biophysically detailed models reproduce perceptron learning \citep{moldwin_perceptron_2020} and solve nonlinear classification that a linear unit cannot \citep{gidon_dendritic_2020, jones_might_2021, jones_biological_2022, aizenbud_what_2026, beniaguev_single_2021}. Our methods confirm that nonlinar dendritic compartmnets increase computational capacity: a single-compartment version of Dendrinet performs at the level of linear discriminant analysis (i.e., near chance) on the PCC task.  The PCC task, with its reliance on channel covariances, therefore cleanly separates the performance of models with and without dendrites. Where we depart from most of this work is in the task itself. Rather than a benchmark, we designed a task whose only discriminative signal lives in the permuted block structure of the input covariance, such that any above-chance performance must to come from second-order computation. \citet{liu_dendritic_2024} show that a single quadratic neuron can realize the Bayes-optimal classifier for Gaussian class-conditional distribution, which sets a ceiling, and \citet{ujfalussy_global_2018} show that sigmoidal branch integration is near-optimal for correlated inputs, which explains why a clustered, nonlinear dendrite is well-suited. The contribution here is not a new upper bound on dendritic capacity but an account of how a biologically grounded neuron can reach it.

\subsection*{Ablation of NMDAR-dependent nonlinearities corroborate experimental findings}

We demonstrated FSCs support computations underlying second order structure classification, and in the process demonstrated that the Poirazi nonlinearity was necessary for good performance on the task (Fig. \ref{fig:benchmark}F). The Poirazi nonlinearity, directly taken from theoretical work \citep{poirazi_pyramidal_2003} has been experimentally validated \citep{polsky_computational_2004}, and is a proxy for the effects of NMDAR-dependent nonlinearities.  Therefore, the model conditions with the Poirazi nonlinearity inactive may mirror experimental analysis of dendritic computation with NMDARs blocked  \citep{sehgal_compartmentalized_2025, kleindienst_activity-dependent_2011}. Inactivating the Poirazi nonlinearity in Dendrinet prevented the formation of FSCs (Fig. \ref{fig:clusters}C) and, conversely, induced dispersal in the preferred-class-context. In addition, we find that inhibitory FSC are formed, which has not been examined experimentally.  Whether the inhibitory synapse FSCs or the induced dispersal of excitatory synapses occur biologicaly and have experimentally testable functions are unknown.

\subsection*{Learned synaptic placement on dendritic trees and implications for connectomics}
Dendrinet learns where to place synapses, not only how strongly they connect. The structural plasticity rule, DeepST, is adapted from deep rewiring, DEEP R \citep{bellec_deep_2018}, and moves synapses between branches under a fixed connectivity budget while gradient descent sets the weights. This separates Dendrinet from other dendritic models, which fix connectivity and learn weights alone \citep{jones_might_2021, jones_biological_2022, chavlis_dendrites_2025}. Our results show that the nonlinearities increased performance depending on whether the weights were fixed or plastic. This suggests that further work testing whether nonlinearities are important for computation should also consider learned connectivity in addition to learned weights.

However, the synaptic weights and locations found in all of this work do not use biologically plausible learning rules, and instead rely entirely on backpropagation.  This was by design as our goal was not to mimic the learning process in the brain but rather to ask if task-optimization, by any means, yields synaptic architectures and distributions that mimic those found in the brain.  A challenge for future work is to demonstrate task optimization using bioplausible rules. Such work would not only demonstrate a synaptic organization and weight distribution is possible for task performance, but that such organization is bio-plausibly learnable.

Learning connectivity also has implications for analysis in EM connectomics. Most connectomic analysis treats a synapse as an edge between two neurons and discards where on the dendritic tree that synapse lands \citep{bazinet_towards_2023}. This is a safe simplification only if the dendritic location has no functional impact. Our ablation results argue otherwise, because shuffling the dendritic position of both excitatory and inhibitory synapses collapses performance.  Now, the data needed to model and test placement-dependent theories are emerging. The MICrONS reconstruction co-registers the visual responses of tens of thousands of cortical neurons with a dense electron-microscopy volume that includes full dendritic trees \citep{bae_functional_2025}, which makes it possible to test if co-tuned inputs share branches in real cortex. Work in the Drosophila looming circuit already takes this step, using electron-microscopy-constrained compartmental models to examine how differently tuned presynaptic populations are arranged across a dendrite and how that arrangement shapes integration \citep{moreno-sanchez_morphology_2024}.  Technologies do not yet exist to examine experimentally how synapse placement and connectivity patterns change during learning.

\subsection*{Inhibitory clustering and circuit analogues}

We found that inhibitory FSCs formed on distal branches in our highest performing models. Dual-color in vivo imaging shows that inhibitory synapse remodeling is not only spatially clustered within roughly ten micrometers and shaped by sensory experience \citep{chen_clustered_2012}, but also that excitatory and inhibitory inputs are co-regulated on individual branches \cite{iascone_whole-neuron_2020}. What is not established is whether inhibitory clusters form by the same correlation-driven logic that builds excitatory ones. The more surprising result is what happened when we removed the NMDAR-dependent-like nonlinearity. Excitatory clustering fell, as expected, but inhibitory clustering rose. This may be compensation, where with supralinear amplification gone, the network can no longer separate classes by boosting on-target excitatory variance, so it leans harder on the lever that remains, which is divisive suppression of off-target variance through clustered inhibition. The compensation is not enough to rescue performance, but it shows that the two clustering populations are not independent and that inhibitory placement is responsive to excitatory drive.

Seeing the importance of inhibitory activity, it becomes relevant that the circuit analogues for Dendrinet remain to be defined. Dendrinet assumes that excitatory and inhibitory inputs arrive at the readout neuron from a shared upstream source, that the inhibition is feedforward, and that optimization placed both onto distal branches. Possible analogues can be found in cortex and hippocampus. In cortex, layer 1 is made up mostly of the distal apical dendrites of deeper pyramidal cells together with a sparse set of interneurons. These layer 1 cells deliver inhibition to local distal dendrites under the control of long-range top-down input \citep{abs_learning-related_2018}. In hippocampus, the temporoammonic pathway carries entorhinal input into stratum lacunosum-moleculare, where it excites the distal dendrites of CA1 pyramidal cells and, through the same afferents, drives neurogliaform interneurons that feed inhibition onto those same distal dendrites \citep{capogna_neurogliaform_2011, price_neurogliaform_2005}. The arrangement of one shared source driving both distal excitation and feedforward distal inhibition mirrors that learned by Dendrinet during task optimization. In both cases, inhibition acts to normalize, which is the dendritic instance of the canonical divisive normalization computation \citep{carandini_normalization_2012}. Despite this, the cortical and hippocampal cases are not equally good analogies. The hippocampal neurogliaform motif matches the feedforward, single-source structure of the model directly, whereas the cortical layer 1 case is a looser fit, because it is largely top-down rather than a feedforward copy of the driving input.

\subsection*{Limitations and Future directions}
Dendrinet treats the dendritic tree as a set of discrete branches separated by depth, and it represents synaptic position as branch membership rather than as continuous distance along a dendrite. This keeps the model tractable and matches the branch-as-subunit abstraction, but it gives up the fine spatial scale at which real clusters form, which is on the order of five to ten micrometers within a branch \citep{takahashi_locally_2012, ujfalussy_impact_2020}. Moving toward continuous position, and toward networks of such neurons, is the natural next step, and single-branch structural-plasticity models offer a starting point for scaling without adding full temporal dynamics \citep{moldwin_gradient_2021}. Two extensions matter most for biological fidelity, one spatial and one temporal. The spatial extension would refine position from branch identity to location within a branch, so that the geometry inside a cluster becomes something the model can represent. The temporal extension matters because the correlations that define a functional cluster are ultimately correlations in time, and a model with subthreshold dynamics could ask how input timing, not only input grouping, shapes which clusters form. We have shown that grouping alone supports covariance computation. Further work investigating whether timing sharpens that result or constrains it could be promising.

The position of a synapse on a dendrite is a functional variable, not a wiring detail. Reading it that way changes how an NMDAR-block deficit should be interpreted, because the block removes nonlinear computation and reorganizes inhibitory clustering alongside any loss of excitatory clusters. It also makes the testable prediction that impairing the dendritic nonlinearity should redistribute clustering toward inhibition rather than abolish it, something dual-color imaging during pharmacological block could resolve \citep{dong_simultaneous_2025}. Whether co-tuned excitatory and inhibitory inputs share distal branches in real tissue is now within reach of co-registered functional and EM data \citep{bae_functional_2025}, and it is the measurement that would tell us whether the placement logic Dendrinet learned is the one biology uses.

\newpage
\newpage

\begin{ack}
This work has been made possible in part by a gift from the Chan Zuckerberg Initiative Foundation to establish the Kempner Institute for the Study of Natural and Artificial Intelligence at Harvard University.
\end{ack}

\medskip

\small
\bibliographystyle{unsrtnat} % Defines the layout (see options below)
\bibliography{clustering_bibliography}    % Points to references.bib (NO .bib extension!)

%%%%%%%%%%%%%%%%%%%%%%%%%%%%%%%%%%%%%%%%%%%%%%%%%%%%%%%%%%%%
\newpage
\appendix

\section{Functional Synapse Cluster benefit proof}

\begin{center}
    \includegraphics[width=0.75 \linewidth]{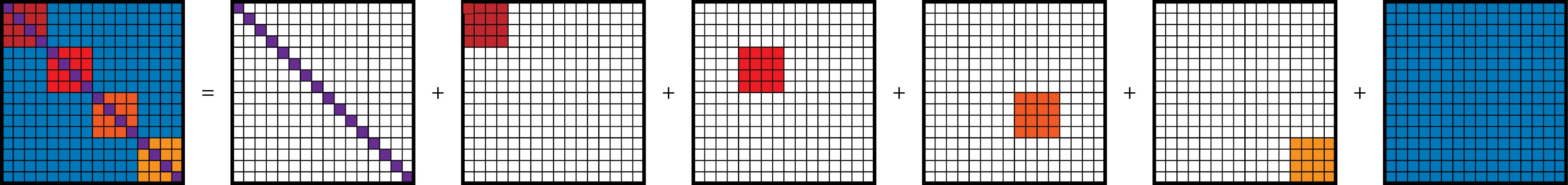}
\end{center}

\begin{equation}\label{eq:cov_decomp}
    \Sigma_k = (1-\rho_l)I_N
    + (\rho_l - \rho_g) \sum_j^{N_b} P_k e^{(j)}e^{(j)^\top} P_k^\top
    + \rho_g \mathbf{1}\mathbf{1}^\top,
\end{equation}

where $e^{(j)} \in \{0, 1 \}^N$ is a binary vector with nonzero entries at the indices satisfying $ \frac{jN}{N_b}\leq i < \frac{(j + 1)N}{N_b}$.

\begin{lemma}[Precision matrix decomposition] \label{lem:inverse_decomposition}
The class-$k$ precision matrix, $\Sigma_k^{-1}$, can be written as:
\begin{gather*}
    \alpha = 1-\rho_l, \qquad
    \beta = 1+(n-1)\rho_l-n\rho_g \\
    \gamma = \frac{1}{\beta(\beta+\rho_gN)} \\
    \Sigma^{-1}_k = \frac{1}{\alpha}I_N - \frac{(\rho_l - \rho_g)}{\alpha \beta}\sum^{N_b}_i \left( P_k e_ie_i^\top P_k^{\top} \right) - \rho_g \gamma \mathbf{1}\mathbf{1}^\top.
\end{gather*}

\end{lemma}

\begin{proof}

\noindent\textbf{Context.}
For a latent $z$ and a candidate cluster configuration $k$, decoding chooses
\[
K^* = \argmin_k \; z^\top \Sigma_k^{-1} z ,
\]
so we need a usable closed form for $\Sigma_k^{-1}$. The covariance for configuration $k$
is a scaled identity plus a sum of block dyads plus a rank-one background term:
\[
\Sigma_k = (1-P_L)I_N
         + (P_L-P_g)\sum_{i=1}^{N_b} e_i^{(k)} e_i^{(k)\top}
         + P_g\,\vone\vone^\top .
\]
 
\medskip
\noindent\textbf{Notation.} Throughout write $e_i \equiv e_i^{(k)}$ and set
\[
\beta = 1-P_L, \qquad \bl = P_L-P_g .
\]
There are $N_b$ blocks of size $n$ that partition $\{1,\dots,N\}$ with $N = N_b\,n$. The
block indicators $e_i\in\{0,1\}^N$ therefore satisfy
\[
e_i^\top e_j = n\,\delta_{ij}, \qquad \sum_{i=1}^{N_b} e_i = \vone,
\qquad e_i^\top \vone = n, \qquad \vone^\top\vone = N .
\]
(The relation $e_i^\top e_j = n\,\delta_{ij}$ is the ``by construction'' orthogonality of the
blocks; the factor $n$ is what produces the $\beta+n\bl$ denominators below.)

\end{proof}

\begin{lemma}
With the notation above, and whenever the scalar denominators are nonzero,
\[
\Sigma_k^{-1}
= \frac{1}{\beta}\,I_N
 - \frac{\bl}{\beta(\beta+n\bl)}\sum_{i=1}^{N_b} e_i e_i^\top
 - \frac{P_g}{(\beta+n\bl)\,(\beta+n\bl+P_g N)}\,\vone\vone^\top .
\]
\end{lemma}
 
\begin{proof}
Introduce the partial covariances
\[
\Sigma^{(0)} = \beta I_N,
\qquad
\Sigma^{(m)} = \beta I_N + \bl\sum_{i=1}^{m} e_i e_i^\top \quad (1\le m\le N_b),
\]
so that $\Sigma_k = \Sigma^{(N_b)} + P_g\,\vone\vone^\top$. We build the inverse by adding one
rank-one term at a time using the Sherman-Morrison identity
\[
(A + uv^\top)^{-1} = A^{-1} - \frac{A^{-1}u\,v^\top A^{-1}}{1 + v^\top A^{-1} u}. \tag{SM}
\]
 
\medskip
\noindent\textbf{Step 1 (block terms, by induction).}
We claim that for every $0\le m\le N_b$,
\[
\big(\Sigma^{(m)}\big)^{-1}
= \frac{1}{\beta}\,I_N
 - \frac{\bl}{\beta(\beta+n\bl)}\sum_{i=1}^{m} e_i e_i^\top . \tag{$\star$}
\]
 
\emph{Base case $m=0$.} The right-hand side is $\tfrac{1}{\beta}I_N = (\beta I_N)^{-1}$. \checkmark
 
\emph{Inductive step.} Suppose $(\star)$ holds at level $m-1$ and denote that matrix by
$A^{-1}$. Because the blocks are disjoint, $e_i^\top e_m = n\,\delta_{im} = 0$ for all
$i\le m-1$, so every dyad in $A^{-1}$ annihilates $e_m$ and
\[
A^{-1} e_m = \frac{1}{\beta}\,e_m,
\qquad
e_m^\top A^{-1} e_m = \frac{1}{\beta}\,e_m^\top e_m = \frac{n}{\beta}.
\]
Since $\Sigma^{(m)} = A + \bl\, e_m e_m^\top$, applying (SM) with $u=v=e_m$ (and coefficient $\bl$) gives
\[
\big(\Sigma^{(m)}\big)^{-1}
= A^{-1} - \frac{\bl\,(A^{-1}e_m)(A^{-1}e_m)^\top}{1 + \bl\, e_m^\top A^{-1} e_m}
= A^{-1} - \frac{(\bl/\beta^2)\,e_m e_m^\top}{1 + n\bl/\beta}
= A^{-1} - \frac{\bl}{\beta(\beta+n\bl)}\,e_m e_m^\top ,
\]
where the last step used $1 + n\bl/\beta = (\beta+n\bl)/\beta$. This is exactly $(\star)$ at
level $m$, completing the induction. In particular,
\[
\big(\Sigma^{(N_b)}\big)^{-1}
= \frac{1}{\beta}\,I_N
 - \frac{\bl}{\beta(\beta+n\bl)}\sum_{i=1}^{N_b} e_i e_i^\top . \tag{1}
\]
 
\medskip
\noindent\textbf{Step 2 (background term).}
Using $\sum_i e_i = \vone$ and $e_i^\top\vone = n$, apply (1) to $\vone$:
\[
\big(\Sigma^{(N_b)}\big)^{-1}\vone
= \frac{1}{\beta}\vone
 - \frac{\bl}{\beta(\beta+n\bl)}\sum_{i=1}^{N_b} e_i\,(e_i^\top\vone)
= \frac{1}{\beta}\vone - \frac{n\bl}{\beta(\beta+n\bl)}\vone
= \frac{1}{\beta+n\bl}\,\vone ,
\]
since $\tfrac{1}{\beta}-\tfrac{n\bl}{\beta(\beta+n\bl)}
= \tfrac{(\beta+n\bl)-n\bl}{\beta(\beta+n\bl)} = \tfrac{1}{\beta+n\bl}$. Consequently
\[
\vone^\top\big(\Sigma^{(N_b)}\big)^{-1}\vone = \frac{\vone^\top\vone}{\beta+n\bl} = \frac{N}{\beta+n\bl}.
\]
Now $\Sigma_k = \Sigma^{(N_b)} + P_g\,\vone\vone^\top$, so a final application of (SM) with
$u=v=\vone$ (coefficient $P_g$) yields
\[
\Sigma_k^{-1}
= \big(\Sigma^{(N_b)}\big)^{-1}
 - \frac{P_g\,\big(\Sigma^{(N_b)}\big)^{-1}\vone\,\vone^\top\big(\Sigma^{(N_b)}\big)^{-1}}
        {1 + P_g\,\vone^\top\big(\Sigma^{(N_b)}\big)^{-1}\vone}
= \big(\Sigma^{(N_b)}\big)^{-1}
 - \frac{\dfrac{P_g}{(\beta+n\bl)^2}\,\vone\vone^\top}{1 + \dfrac{P_g N}{\beta+n\bl}} .
\]
The scalar prefactor simplifies to
\[
\frac{P_g/(\beta+n\bl)^2}{1 + P_g N/(\beta+n\bl)}
= \frac{P_g}{(\beta+n\bl)\,(\beta+n\bl+P_g N)} ,
\]
and substituting (1) for $\big(\Sigma^{(N_b)}\big)^{-1}$ gives the stated expression.

\end{proof}

\begin{lemma}[QDA necessary average condition]\label{lem:qda_snr}
Under the factorized covariance model with equal parameters $(\rho_l, \rho_g, n)$ across classes, and also equal class means and priors, with $z$ centered by the common mean, a necessary condition for correct classification of $z$ to class $k^*$ under the QDA decision rule requires:
\begin{gather*}
    \sigma_{\mathrm{avg}}(z,k)\triangleq
    \frac{(K-1)\sum_i z^\top P_ke_ie_i^\top P_k^\top z}
    {\sum_{k'\neq k}\sum_i z^\top P_{k'}e_ie_i^\top P_{k'}^\top z}\\
    \sigma_{\mathrm{avg}}(z,k^*)>1.
\end{gather*}

\end{lemma}

% Tighter proof
\begin{proof}
The QDA decision rule assigns $z$ to $k^*$ when: 
\begin{equation*}
z^\top\Sigma_{k^*}^{-1}z + \log|\Sigma_{k^*}| < z^\top\Sigma_{k'}^{-1}z + \log|\Sigma_{k'}|, \quad \forall k'\neq k^*.
\end{equation*}
Summing over all $k' \neq k^*$ and rearranging:
\begin{equation*}
(K-1)\left(z^\top\Sigma_{k^*}^{-1}z + \log|\Sigma_{k^*}|\right) < \sum_{k'\neq k^*}\left(z^\top\Sigma_{k'}^{-1}z + \log|\Sigma_{k'}|\right).
\end{equation*}
Since all classes have equal parameters $(\rho_l, \rho_g, n)$, their covariance matrices have identical eigenvalues and therefore equal determinants, so $\log|\Sigma_{k}| = \log|\Sigma_{k'}|$ for all $k, k'$. These log determinant terms cancel leaving:
\begin{equation*}
    z^\top\Sigma_{k^*}^{-1}z < \frac{1}{K-1}\sum_{k'\neq k^*}z^\top\Sigma_{k'}^{-1}z.
\end{equation*}
Substituting Lemma~\ref{lem:inverse_decomposition} and rearranging terms gives:
\begin{equation*}
    \frac{(K-1)\left(\frac{1}{\alpha}\|z\|^2 - \frac{(\rho_l-\rho_g)}{\alpha\beta}\sum_i z^\top P_{k^*}e_ie_i^\top P_{k^*}^\top z - \rho_g\gamma z^\top\mathbf{1}\mathbf{1}^\top z\right)}{\sum_{k'\neq k^*}\left(\frac{1}{\alpha}\|z\|^2 - \frac{(\rho_l-\rho_g)}{\alpha\beta}\sum_i z^\top P_{k'}e_ie_i^\top P_{k'}^\top z - \rho_g\gamma z^\top\mathbf{1}\mathbf{1}^\top z\right)} < 1.
\end{equation*}
Since the terms $\frac{1}{\alpha}\|z\|^2$ and $\rho_g \gamma z^\top \mathbf{1}\mathbf{1}^\top z$ are class-invariant, they appear identically in every term in the numerator and denominator the inequality is equivalent to:
\[
-\frac{\rho_l-\rho_g}{\alpha\beta}\sum_i z^\top P_{k^*}e_ie_i^\top P_{k^*}^\top z
< -\frac{\rho_l-\rho_g}{\alpha\beta(K-1)}
\sum_{k'\neq k^*}\sum_i z^\top P_{k'}e_ie_i^\top P_{k'}^\top z.
\]

Since $\rho_l > \rho_g$ and $\alpha, \beta > 0$, we have $\frac{\alpha \beta}{(\rho_l - \rho_g)} > 0$. Multiplying both sides by $-\frac{\alpha \beta}{(\rho_l - \rho_g)}$ and flipping the inequality yields:
\begin{equation*}
    \sum_i z^\top P_{k^*}e_ie_i^\top P_{k^*}^\top z > \frac{1}{K-1}\sum_{k'\neq k^*}\sum_i z^\top P_{k'}e_ie_i^\top P_{k'}^\top z.
\end{equation*}
Collecting all terms on one side gives:
\begin{equation}
    \sigma_{\mathrm{avg}}(z, k^*) = \frac{(K-1)\sum_i z^\top P_{k^*}e_ie_i^\top P_{k^*}^\top z}{\sum_{k'\neq k^*}\sum_i z^\top P_{k'}e_ie_i^\top P_{k'}^\top z} > 1.
\end{equation}
\end{proof}

\begin{lemma}[Variance of a neuron]\label{lem:single_var}
\[
\operatorname{Var}[\phi(W^\top z)\mid k]
\approx \big[\phi'(W^\top\bar z_k)\big]^2W^\top\Sigma_kW.
\]
\end{lemma}
\begin{proof}
\noindent
Let $y(z)$ be the neuron output, linearized to first order about the mean
latent $\bar{z}$:
\begin{align}
y(z) &= \phi(w^\top z) \\[4pt]
y(z) &\approx \phi(w^\top \bar{z}) + \phi'(w^\top \bar{z})\, w^\top (z - \bar{z}).
\end{align}

\noindent
Definitions:
\begin{equation}
\Var(z) = \Sigma_k, \qquad
\bar{z} = \E[z], \qquad
\Sigma_k = \E\!\left[(z-\bar{z})(z-\bar{z})^\top\right].
\end{equation}

\noindent
The variance of the output is
\begin{equation}
\Var\big(y(z)\big) = \E\!\left[y(z)^\top y(z)\right] - \E[y(z)]^\top\,\E[y(z)].
\end{equation}

\begin{align}
\E[y(z)]
  &= \E\!\left[\phi(w^\top \bar{z})\right]
   + \E\!\left[\phi'(w^\top \bar{z})\, w^\top (z - \bar{z})\right] \\
  &= \E\!\left[\phi(w^\top \bar{z})\right]
   + \phi'(w^\top \bar{z})\, w^\top \big(\E[z] - \bar{z}\big) \\[4pt]
\E[y(z)] &= \phi(w^\top \bar{z}),
\end{align}
since $\E[z] - \bar{z} = 0$.

\begin{align}
\E[y^\top y]
  &= \E\!\Big[\big(\phi(w^\top \bar{z}) + \phi'(w^\top \bar{z})\, w^\top (z-\bar{z})\big)^\top
              \big(\phi(w^\top \bar{z}) + \phi'(w^\top \bar{z})\, w^\top (z-\bar{z})\big)\Big]
\end{align}
The cross (linear-in-$z$) terms vanish under expectation, leaving
\begin{align}
\E[y^\top y]
  &= \E\!\left[\phi^\top(w^\top \bar{z})\,\phi(w^\top \bar{z})\right]
   + \E\!\Big[\underbrace{\big(\phi'(w^\top \bar{z})\, w^\top (z-\bar{z})\big)^\top
              \big(\phi'(w^\top \bar{z})\, w^\top (z-\bar{z})\big)}_{\text{scalar, so we can rearrange}}\Big] \\[4pt]
  &= \phi^\top(w^\top \bar{z})\,\phi(w^\top \bar{z})
   + \E\!\left[\phi'(w^\top \bar{z})\, w^\top (z-\bar{z})(z-\bar{z})^\top w\, \phi'(w^\top \bar{z})^\top\right] \\[4pt]
  &= \phi^\top(w^\top \bar{z})\,\phi(w^\top \bar{z})
   + \phi'(w^\top \bar{z})\, w^\top\, \E\!\left[(z-\bar{z})(z-\bar{z})^\top\right] w\, \phi'^\top(w^\top \bar{z}).
\end{align}
Substituting $\Sigma_k = \E[(z-\bar{z})(z-\bar{z})^\top]$,
\begin{equation}
\E[y^\top y]
  = \phi^\top(w^\top \bar{z})\,\phi(w^\top \bar{z})
  + \phi'(w^\top \bar{z})\, w^\top \Sigma_k\, w\, \phi'^\top(w^\top \bar{z}).
\end{equation}

\bigskip
\noindent\textbf{Result.}
Subtracting $\E[y]^\top\E[y] = \phi^\top(w^\top\bar{z})\,\phi(w^\top\bar{z})$ leaves
\begin{equation}
\Var(y \mid k) \approx\;
\phi'(w^\top \bar{z})\, w^\top \Sigma_k\, w\, \phi'^\top(w^\top \bar{z}).
\end{equation}

\end{proof}

\begin{lemma}[Shunting Variance Approximation]\label{lem:shunting_var}
The variance of a shunting neuron's output $y=\phi\left(\frac{E}{1+E+I}\right)$ for a leaf branch in Eq.~\ref{eq:vinf} is approximately:
Let $\bar V=\bar E/(1+\bar E+\bar I)$. Then
\[
\operatorname{Var}[y]\approx
\frac{\phi'(\bar V)^2}{(1+\bar E+\bar I)^4}
\left((1+\bar I)^2\operatorname{Var}[E]+\bar E^2\operatorname{Var}[I]
-2\bar E(1+\bar I)\operatorname{Cov}[E,I]\right).
\]
\end{lemma}

\begin{proof}
The activity of the shunting neuron is
\[
y=\phi(V), \qquad V=\frac{E}{1+E+I}.
\]
Using the first-order delta method around $(\bar E,\bar I)$,
\[
\operatorname{Var}[y]\approx
\begin{pmatrix}
\pdv{y}{E} & \pdv{y}{I}
\end{pmatrix}
\begin{pmatrix}
\operatorname{Var}[E] & \operatorname{Cov}[E,I]\\
\operatorname{Cov}[E,I] & \operatorname{Var}[I]
\end{pmatrix}
\begin{pmatrix}
\pdv{y}{E}\\
\pdv{y}{I}
\end{pmatrix},
\]
where the derivatives are evaluated at $(\bar E,\bar I)$.

By the chain rule,
\[
\pdv{y}{E}
=
\phi'(V)\frac{1+I}{(1+E+I)^2},
\qquad
\pdv{y}{I}
=
-\phi'(V)\frac{E}{(1+E+I)^2}.
\]
Let
\[
\bar V=\frac{\bar E}{1+\bar E+\bar I}.
\]
Evaluating the derivatives at the mean gives
\[
\operatorname{Var}[y]\approx
\frac{\phi'(\bar V)^2}{(1+\bar E+\bar I)^4}
\left(
(1+\bar I)^2\operatorname{Var}[E]
+\bar E^2\operatorname{Var}[I]
-2\bar E(1+\bar I)\operatorname{Cov}[E,I]
\right).
\]
\end{proof}

\synapticclustering*

% Thight proof

\begin{proof}
Under equal class means and a common nonzero scalar derivative, the delta-method derivative factor cancels from the variance ratio. We assume that the nonzero derivative is guaranteed by our loss function.

Given
\[
W^\top\mathbf 1=s,\qquad
\|W\|^2=s.
\]

Let
\[
B_k=\sum_iP_ke^{(i)}e^{(i)\top}P_k^\top,\qquad
q_k(W)=W^\top B_kW.
\]
Using the covariance in Eq.~\ref{eq:cov_decomp},
\[
W^\top\Sigma_kW=(1-\rho_l)s+\rho_gs^2+(\rho_l-\rho_g)q_k(W).
\]
If $c_b$ is the number of selected synapses in target block $b$, then $q_k(W)=\sum_bc_b^2$. Thus
\[
q_k(W_{\rm cluster})=s^2,\qquad
q_k(W_{\rm split})=(s-1)^2+1,
\]
and the target variance difference is $2(s-1)(\rho_l-\rho_g)$.

For a uniform random competing partition, two selected inputs share a block with probability $(n-1)/(N-1)$. Therefore
\[
\mathbb E_{k'}[q_{k'}(W)]
=s+s(s-1)\frac{n-1}{N-1},
\]
which depends on $s$ but not on the target placement. This gives the denominator $D_s$ and the gap stated in the theorem. Also $\sum_bc_b^2$ is maximal at $s^2$ when all synapses are in one block, provided $s\leq n$.

The equality is only for an expected competitor; with three competing classes concentration is not guaranteed. The normalized gap is positive, but it may decrease as $s$ grows.
\end{proof}

\begin{theorem}[Excitatory-Inhibitory Diversity]\label{thm:shunting_ei}
For a leaf branch, fixed excitatory and inhibitory synapse counts with $m_e+m_i\leq n$, binary disjoint supports, and each population already clustered in one target block, the expected-competitor SNR is higher when the two populations use different target blocks than when they use the same block.
\end{theorem}
\begin{proof}
Let
\[
d=(1+\bar I)W_e-\bar E W_i.
\]
Apart from a common positive delta-method factor, the leaf variance is $d^\top\Sigma_kd$.

Let $a=1+\bar I$ and $b=\bar E$. If excitation and inhibition use the same target block, the block term is $(am_e-bm_i)^2$. If they use different target blocks, it is $(am_e)^2+(bm_i)^2$. The second case is larger by
\[
2abm_em_i=2(1+\bar I)\bar E m_em_i>0.
\]
After multiplying by $\rho_l-\rho_g>0$, target variance is larger. Under the stated disjoint-support and random-competitor assumptions, the expected competitor term is the same in both cases.

This is only a comparison of these two placements. It does not prove a global maximum for the full dendritic tree.
\end{proof}

%%%%%%%%%%%%%%%%%%%%%%%%%%%%%%%%%%%%%%%%%%%%%%%%%%%%%%%%%%%%

\newpage

\end{document}